\documentclass[10pt,onecolumn,draftcls,journal]{IEEEtran}
\usepackage[ansinew]{inputenc}
\usepackage[pdftex]{graphicx}
\usepackage[T1]{fontenc}
\usepackage{color,amsmath,enumerate,amssymb,hhline,verbatim}
\usepackage{url}
\usepackage{multirow}

\DeclareMathOperator{\rank}{rank}
\begin{document}

\title{Linear Locally Repairable Codes with Random Matrices
\thanks{Part of this work appeared at Global Wireless Summit 2014, Aalborg, Denmark.}
}

\author{Toni Ernvall, Thomas Westerb{\"a}ck and  Camilla Hollanti
\thanks{T. Ernvall is with Turku Centre for Computer Science, Turku, Finland and with the Department of Mathematics and Statistics, University of Turku, Finland (e-mail:tmernv@utu.fi).}

\thanks{T. Westerb{\"a}ck is with Department of Mathematics and Systems Analysis, Aalto University (e-mail:thomas.westerback@aalto.fi).}

\thanks{C. Hollanti is with Department of Mathematics and Systems Analysis, Aalto University (e-mail:camilla.hollanti@aalto.fi).}
}

\maketitle

\newtheorem{definition}{Definition}[section]
\newtheorem{thm}{Theorem}[section]
\newtheorem{proposition}[thm]{Proposition}
\newtheorem{lemma}[thm]{Lemma}
\newtheorem{corollary}[thm]{Corollary}
\newtheorem{exam}{Example}[section]
\newtheorem{conj}{Conjecture}
\newtheorem{remark}{Remark}[section]

\newcommand{\La}{\mathbf{L}}
\newcommand{\h}{{\mathbf h}}
\newcommand{\Z}{{\mathbf Z}}
\newcommand{\R}{{\mathbf R}}
\newcommand{\C}{{\mathbf C}}
\newcommand{\D}{{\mathcal D}}
\newcommand{\F}{{\mathbf F}}
\newcommand{\HH}{{\mathbf H}}
\newcommand{\OO}{{\mathcal O}}
\newcommand{\G}{{\mathcal G}}
\newcommand{\A}{{\mathcal A}}
\newcommand{\B}{{\mathcal B}}
\newcommand{\I}{{\mathcal I}}
\newcommand{\E}{{\mathcal E}}
\newcommand{\PP}{{\mathcal P}}
\newcommand{\Q}{{\mathbf Q}}
\newcommand{\M}{{\mathcal M}}
\newcommand{\separ}{\,\vert\,}
\newcommand{\abs}[1]{\vert #1 \vert}

\begin{abstract}
In this paper, locally repairable codes with all-symbol locality are studied. Methods to modify already existing codes are presented. Also, it is shown that with high probability, a random matrix with a few extra columns guaranteeing the locality property, is a generator matrix for a locally repairable code with a good minimum distance. The proof of this gives also a constructive method to find locally repairable codes.
\end{abstract}

\section{Introduction}
\subsection{Locally Repairable Codes}
In the literature, three kinds of repair cost metrics are studied: \emph{repair bandwidth} \cite{dimakis}, \emph{disk-I/O} \cite{diskIO}, and \emph{repair locality} \cite{Gopalan,Oggier,Simple}. In this paper, the repair locality is the subject of interest.

Given a finite field $\mathbb{F}_q$ with $q$ elements and an injective function $f: \mathbb{F}_q^k \rightarrow \mathbb{F}_q^n$, let $C$ denote the image of $f$. We say that $C$ is a \emph{locally repairable code (LRC)} and has \emph{all-symbol $(r,\delta)$-locality} with parameters $(n,k,d)$, if the code $C$ has minimum (Hamming) distance $d$ and all the $n$ symbols of the code have $(r,\delta)$-locality. The concept was introduced in \cite{prakash}. The $j$th symbol has $(r,\delta)$-locality if there exists a subset $S_j \subseteq \{1,\dots,n\}$ such that $j \in S_j$, $|S_j| \leq r+\delta-1$ and the minimum distance of the code obtained by deleting code symbols corresponding the elements of $\{1,\dots,n\} \setminus S_j$ is at least $\delta$. LRCs are defined when $1 \leq r \leq k$. By a linear LRC we mean a linear code of length $n$ and dimension $k$.

In \cite{prakash} it is shown that we have the following bound for a locally repairable code $C$ of length $n$, dimension $k$, minimum distance $d$ and all-symbol $(r,\delta)$-locality:
\begin{equation}
\label{upperbound}
d \leq n - k - \left( \left \lceil \frac{k}{r} \right \rceil -1 \right)\left( \delta - 1 \right) + 1
\end{equation}
A locally repairable code that meets this bound is called \emph{optimal}. For this reason we write $d_{\text{opt}}(n,k,r,\delta)=n - k - \left( \left \lceil \frac{k}{r} \right \rceil -1 \right)\left( \delta - 1 \right) + 1$.

\subsection{Related Work}
In the all-symbol locality case the information theoretic trade-off between locality and code distance for any (linear or nonlinear) code was derived in \cite{LRCpapailiopoulos}. In \cite{LRCmatroid}, \cite{SongOptimal}, \cite{Rawat} and \cite{TamoBarg} the existence of optimal LRCs was proved for several parameters $(n,k,r)$. Good codes with the weaker assumption of information symbol locality are designed in \cite{Pyramid}. In \cite{Gopalan} it was shown that there exist parameters $(n,k,r)$ for linear LRCs for which the bound of Eq. \eqref{upperbound} is not achievable. LRCs corresponding MSR and MBR points are studied in \cite{Kamath}.

\subsection{Contributions and Organization}
In this paper, we will study codes with all-symbol locality, when given parameters $n$, $k$, $r$, and $\delta$. We will show methods to find smaller and larger codes when given a locally repairable code. At some occasions when the starting point is optimal, also the resulting code is optimal. We will also show that random matrices with a few non-random extra columns guaranteeing the repair property generate a linear LRCs with good minimum distance, with probability approaching to one as the field size approaches the infinity.

Section \ref{Sec:BuildingCodes} gives two procedures to exploit already existing codes when building new ones. In that section we are restricted to the case $\delta=2$. To be exact, it explains how we can build a new linear code of length $n+1$ and dimension $k+1$ with all-symbol repair locality $r+1$ from an already existing linear code of length $n$ and dimension $k$ with all-symbol repair locality $r$ such that the minimum distance remains the same. The same section also introduces a method to find a smaller code when given a code associated to parameters $(n,k,r)$. Namely the procedure gives a code of length $n-1$, dimension $k' \geq k-1$, minimum distance $d' \geq d$ and all-symbol locality $r$.

In Section \ref{Sec:Random} we study random matrices with a few non-random extra columns guaranteeing the repair property. By using the ideas of Section \ref{Sec:minDistance} it is shown that these random codes perform well with high probability. The proof of this is postponed to Section \ref{Sec:minDistance}.

In Section \ref{Sec:minDistance} we give a construction of almost optimal linear locally repairable codes with all-symbol locality. By almost optimal we mean that the minimum distance of a code is at least $d_{\text{opt}}(n,k,r,\delta)-\delta+1$.

\section{Building Codes from Other Codes}\label{Sec:BuildingCodes}
\subsection{Definitions}
In this section we will restrict us in the case $\delta=2$ and show how we can build a new linear code of length $n+1$ and dimension $k+1$ with all-symbol repair locality $r+1$ from an already existing linear code of length $n$ and dimension $k$ with all-symbol repair locality $r$ such that the minimum distance remains the same. Also, we will show how to find a code for parameters $(n'=n-1,k' \geq k-1,d' \geq d,r'=r)$. That is, we will show how to enlarge codes and how to reduce codes.

First we need some definitions. Here $q$ is a prime power and $\mathbb{F}_q$ is a finite field with $q$ elements. Let $\mathbf{x},\mathbf{y} \in \mathbb{F}_q^n$. Then $d(\mathbf{x},\mathbf{y})$ is the Hamming distance of $\mathbf{x}$ and $\mathbf{y}$. The weight of $\mathbf{x}$ is $w(\mathbf{x})=d(\mathbf{x},\mathbf{0})$.
The sphere with radius $s$ and center $\mathbf{x}$ is defined as
\[
B_s(\mathbf{x})=\{\mathbf{y} \in \mathbb{F}_q^n \mid d(\mathbf{x},\mathbf{y}) \leq s\}.
\]
Define furthermore
\[
V_q(n,s)=|B_s(\mathbf{x})|= \sum_{i=0}^{s} \binom{n}{i} (q-1)^i.
\]
For $V_q(n,s)$ we have a trivial upper bound
\[
V_q(n,s) \leq (1+s) \binom{n}{\lfloor \frac{n}{2} \rfloor} q^s.
\]

\subsection{Enlarging codes}
If $r=k$ then we always get an optimal linear LRC by an MDS code. Hence in this section we will assume that $r<k$.
\begin{thm}\label{Thm:Enlarge}
Suppose we have a linear LRC for parameters $(n,k,d,r)$ over a field $\mathbb{F}_q$ with $q > d \binom{n}{\left\lfloor \frac{n}{2} \right\rfloor}$ and $r<k$. Then there exists a linear LRC for parameters $(n'=n+1,k'=k+1,d'=d,r'=r+1)$ over the same field.
\end{thm}
\begin{proof}
Let $C$ be a linear LRC for parameters $(n,k,d,r)$ over a field $\mathbb{F}_q$ with $q > d \binom{n}{\left\lfloor \frac{n}{2} \right\rfloor}$. Let $G$ be its generator matrix, \emph{i.e.}, $G$ is $k \times n$ matrix such that its row vectors form a basis for $C$.

Suppose that $k$ is maximal in the meaning that there does not exist a linear LRC for parameters $(n,k+1,d,r)$ over a field $\mathbb{F}_q$. This assumption can be made without loss of generality because we can remove extra base vectors from the resulting code if the dimension is too large. This does not reduce the minimum distance nor increase the repair locality.

Notice first that (\ref{upperbound}) gives
\[
k + d \leq n-\left\lceil\frac{k}{r}\right\rceil+2 \leq n-2+2 = n
\]
and hence
\begin{equation}
\begin{split}
|C|V_q(n,d-1) & \leq q^k \cdot (1+d-1) \binom{n}{\left\lfloor \frac{n}{2} \right\rfloor} q^{d-1} \\
& = d \binom{n}{\left\lfloor \frac{n}{2} \right\rfloor} q^{k+d-1} \\
& < q^{k+d} \\
& \leq q^{n}.
\end{split}
\end{equation}

Therefore there exists a vector $\mathbf{x} \in \mathbb{F}_q^n$ with distance at least $d$ to vectors of $C$. Denote by $G_2$ a new $(k+1) \times (n+1)$ matrix
\[
\left(
\begin{array}{c|c}
G & \mathbf{0}^t \\ \hline
\mathbf{x} & 1
\end{array}
\right)
\]
where $\mathbf{0}$ is an all-zero vector from $\mathbb{F}_q^k$.

Denote by $C_2$ a code that $G_2$ generates. Clearly $C_2 \subseteq \mathbb{F}_q^{n+1}$ and its dimension is $k+1$. Its minimum distance is $d$: Let $\mathbf{u}=a\mathbf{y}+\mathbf{z}$ where $a \in \mathbb{F}_q$, $\mathbf{y}=(\mathbf{x}|1)$ and $\mathbf{z}=(\mathbf{z'}|0)$ with $\mathbf{z'}$ being a vector form $C$. Now if $a=0$ we have
\[
w(\mathbf{u})=w(\mathbf{z})=w(\mathbf{z'}) \geq d
\]
and if $a \neq 0$ we have
\[
w(\mathbf{u})=w(a^{-1}\mathbf{u})=w(\mathbf{y}+a^{-1}\mathbf{z})=d(\mathbf{y},-a^{-1}\mathbf{z})=d(\mathbf{x},-a^{-1}\mathbf{z'})+d(1,0) \geq d+1.
\]

The code $C_2$ has repair locality $r+1$: In the matroid $\mathcal{M}_1$ represented by $G$ each $i=1,\dots,n$ is contained in a circuit of size at most $r+1$. Hence in the matroid $\mathcal{M}_2$ represented by $G_2$ each $i=1,\dots,n$ is contained in a circuit of size at most $r+2$. And hence we chose $k$ to be maximal we have that $n+1$ is contained in some circuit of size at most $r+2$. This completes the proof.
\end{proof}

The following example illustrates the strength of the above result in the case that $r$ and $k$ are close enough to each other.

\begin{exam}
Suppose $\delta=2$ and let $r \in [\frac{k}{2},k)$ and $C$ be an optimal linear locally repairable code for parameters $(n,k,d,r)$ over a field $\mathbb{F}_q$ with $q > d \binom{n}{\lfloor \frac{n}{2} \rfloor}$. Because of the optimality we have
\[
d = n-k-\left\lceil\frac{k}{r}\right\rceil+2 = n-k.
\]

Theorem \ref{Thm:Enlarge} results a locally repairable code for parameters $(n'=n+1,k'=k+1,d'=d,r'=r+1)$. This code is also optimal:
\[
n'-k'-\left\lceil\frac{k'}{r'}\right\rceil+2=n-k-\left\lceil\frac{k+1}{r+1}\right\rceil+2 = n-k=d=d'.
\]
Hence the proof of the above theorem gives a procedure to build optimal codes using already known optimal codes in the case that the size of the repair locality is at least half of the code dimension.
\end{exam}

\subsection{Puncturing codes}
Puncturing is a traditional method in classical coding theory. The next theorem shows that this method is useful also in the context of locally repairable codes.

\begin{thm}\label{Thm:codePuncturing}
Suppose we have a locally repairable code $C \subseteq \mathbf{F}_{q}^{n}$ with all-symbol locality associated to parameters $(n,k,d,r)$. There exists a code $C' \subseteq \mathbf{F}_{q}^{n-1}$ associated to parameters $(n'=n-1,k',d',r'=r)$ with $k' \geq k-1$ and $d' \geq d$. Also, if $C$ is linear then we may assume that $C'$ is linear.
\end{thm}
\begin{proof}
Write
\[
C_x=\{\mathbf{y} \in C \mid \mathbf{y}=(x,\mathbf{z}) \text{ where } \mathbf{z} \in \mathbf{F}_{q}^{n-1} \}
\]
for $x \in \mathbf{F}_{q}$.

Clearly each element of $C$ is contained in precisely one of the subsets $C_x$ with $x \in \mathbf{F}_{q}$. Hence there exists $a \in \mathbf{F}_{q}$ such that
\[
|C_a| \geq \frac{|C|}{q}=q^{k-1}.
\]
Define $C'$ to be a code we get by puncturing the first component of $C_a$, \emph{i.e.},
\[
C' =\{\mathbf{z} \in \mathbf{F}_{q}^{n-1} \mid (a,\mathbf{z}) \in C_a \}.
\]

Now, $C'$ is of size at least $k$, its minimum distance $d'$ is the same as the minimum distance of $C_a$, that is, at least the minimum distance of $C$ and hence $d' \geq d$, and $C'$ has repair locality $r$. Indeed, suppose we need to repair the $j$th node. If the first node from the original system is not in the repair locality, then the repair can be made as in the original code. If the first node is in the repair locality, then we know that there is $a$ stored into that node and hence the repair can be made using the other nodes from the original locality.

If $C$ is linear then it is easy to verify that we can choose $a$ to be $0$ and in this case also $C'$ is linear.
\end{proof}

\begin{exam}\label{Exam:InductionPuncture}
Suppose that $C$ is an optimal code. It is associated with parameters $(n,k,d,r)$ with equality
\[
d = n-k-\left\lceil\frac{k}{r}\right\rceil +2.
\]
Let $C'$ be a code formed from $C$ using method explained in Theorem \ref{Thm:codePuncturing}. Hence it is associated with parameters $(n'=n-1,k' \geq k-1,d' \geq d,r' =r)$. This code is optimal if
\[
d = n-k-\left\lceil\frac{k-1}{r}\right\rceil +2.
\]
This is true if and only if
\[
\left\lceil\frac{k}{r}\right\rceil=\left\lceil\frac{k-1}{r}\right\rceil,
\]
\emph{i.e.},
$r$ does not divide $k-1$.
\end{exam}

\section{Random matrices as generator matrices for locally repairable codes}\label{Sec:Random}

\subsection{The structure of the codes}
We will study such linear codes that nodes are divided into such non-overlapping sets $S_1,S_2,\dots,S_A$ that any node $x \in S_j$ can be repaired by any $|S_j \setminus \{x\}|-(\delta-2)=|S_j|-\delta+1$ nodes from the set $S_j$. We also require that $|S_j| \leq r+\delta-1$ and to guarantee the all-symbol repairing property, that $\cup_{j=1}^{A} S_j =\{1,\dots,n\}$. Suppose we have a $k$-dimensional linear code and nodes, say, $1,2,\dots,s$ ($\delta \leq s \leq r+\delta-1$) corresponding columns in the generator matrix, form a repair set $S_1$. Denote the $k \times s$ matrix these columns define by $G$ and write $t=s-\delta+1$. Intuitively it is natural to require that $G$ is of maximal rank, that is, the rank of $G$ is $t$.

By the locality assumption, any $t$ columns can repair any other column, \emph{i.e.}, any $t$ columns span the same subspace as all the $s$ columns. So we have
\[
G=(\mathbf{x}_1|\dots|\mathbf{x}_t|\mathbf{y}_{1}|\dots|\mathbf{y}_{\delta-1})
\]
where each $\mathbf{y}_j$ can be represented as a linear combination of $\mathbf{x}_1,\dots,\mathbf{x}_t$ and $\mathbf{x}_1,\dots,\mathbf{x}_t$ are linearly independent. This gives that
\[
G=(\mathbf{x}_1|\dots|\mathbf{x}_t)(I_t|B)
\]
where $I_t$ is an identity matrix of size $t$ and $B$ is $t \times (\delta-1)$ matrix.

Let $G'$ consist of some $t$ columns of $G$ and $C$ consist of the corresponding columns of $(I_t|B)$. It is easy to verify that
\[
G'=(\mathbf{x}_1|\dots|\mathbf{x}_t)C
\]
and hence
\[
\rank(C)=\rank((\mathbf{x}_1|\dots|\mathbf{x}_t)C)=\rank(G')=t.
\]

Consider a submatrix of $B$ corresponding rows $i_1,\dots,i_l$ and columns $j_1,\dots,j_l$. It is easy to check that this submatrix is invertible if and only if a submatrix corresponding the columns $\{1,\dots,t\} \setminus \{i_1,\dots,i_l\}$ and $\{t+j_1,\dots,t+j_l\}$ of $(I_t|B)$ is invertible. This is invertible since the rank of the submatrix of $G$ corresponding the same columns is $t$. Hence any square submatrix of $B$ is invertible

Suppose that matrices $(I_{t_1}|B_1),\dots,(I_{t_A}|B_A)$ are of this form. It is natural to study codes with generator matrix of form
\[
\left((\mathbf{x}_{1,1}|\dots|\mathbf{x}_{1,t_1})(I_{t_1}|B_1) | \dots | (\mathbf{x}_{A,1}|\dots|\mathbf{x}_{A,t_A})(I_{t_A}|B_A) \right)
\]
and ask how we should choose the vectors $\mathbf{x}_{1,1},\dots,\mathbf{x}_{1,t_1},\dots,\mathbf{x}_{A,1},\dots,\mathbf{x}_{A,t_A}$ such that the given code has the biggest possible minimum distance.

Notice also that since the rank of a generator matrix is $k$, we have
\begin{equation}
\begin{split}
k & \leq \rank\left( (\mathbf{x}_{1,1}|\dots|\mathbf{x}_{1,t_1})(I_{t_1}|B_1) \right) + \dots + \rank\left( (\mathbf{x}_{A,1}|\dots|\mathbf{x}_{A,t_A})(I_{t_A}|B_A) \right) \\
& \leq t_1 + \dots + t_A
\end{split}
\end{equation}
and hence
\[
k \leq n -A(\delta-1) \leq n -\left\lceil\frac{n}{r+\delta-1}\right\rceil(\delta-1).
\]

\subsection{Random codes}\label{Subsec:RandomResult}
In this subsection we study locally repairable codes generated by random matrices with a few extra columns consisting of linear combinations of the previous columns guaranteing the repair property. It is shown that this kind of code has a good minimum distance with probability approaching to $1$ as the field size $q$ approaches infinity. The proof of this is postponed to the subsection \ref{Subsec:RandomProof}.

\begin{thm}
Given parameters $(n,k,r,\delta)$ and $A>0$ with $r<k$, $n-A(\delta-1) \geq k$, and positive integers $s_1 \leq s_2 \leq \dots \leq s_A$ such that $n=\sum_{j=1}^{A}s_j$ and $\delta \leq |s_j| \leq r+\delta-1$ for all $j=1,\dots,A$. Assume also we have matrices $B_1,B_2,\dots,B_A$ such that $B_j$ is such a $(s_j-\delta+1)\times (\delta-1)$ matrix that all its square submatrices are invertible ($j=1,\dots,A$). Let $x_{i,j}$ be uniformly independent and identically distributed random variables over $\mathbf{F}_q$.

Consider matrices $E$, $F$ and $G$ that are defined as follows:
\begin{equation}\label{Eq:matrixE}
E=\begin{pmatrix}
  x_{1,1} & x_{1,2} & \cdots & x_{1,n-A(\delta-1)} \\
  x_{2,1} & x_{2,2} & \cdots & x_{2,n-A(\delta-1)} \\
  \vdots  & \vdots  & \ddots & \vdots    \\
  x_{k,1} & x_{k,2} & \cdots & x_{k,n-A(\delta-1)}
 \end{pmatrix}=(E_1|E_2|\dots|E_A),
\end{equation}
where $E_j$ is a $k \times (s_j-\delta+1)$ matrix,
\[
F=(E_1B_1|E_2B_2|\dots|E_AB_A)
\]
and
\[
G=(E|F).
\]
With probability approaching to one as $q \rightarrow \infty$, $G$ is a generator matrix for a $k$-dimensional locally repairable code of length $n$ with all-symbol $(r,\delta)$-locality and minimum distance
\[
d \geq n-k-z(\delta-1)+1
\]
where $z$ is an integer with properties
\[
\sum_{j=1}^{z} (s_j-\delta+1) \leq k-1 \text{ and } \sum_{j=1}^{z+1} (s_j-\delta+1) > k-1.
\]
\end{thm}

\section{Code construction}\label{Sec:minDistance}

\subsection{Construction}\label{Subsec:Construction}

In this subsection we will give a construction for linear locally repairable codes with all-symbol locality over a field $\mathbb{F}_q$ with $q>2(r\delta)^{r4^{r+1}+4^{r+1}}\binom{n+2 (r\delta)^{r4^{r+1}}}{k-1}$ when given parameters $(n,k,r,\delta)$ such that $n-\left\lceil\frac{n}{r+\delta-1}\right\rceil(\delta-1) \geq k$. We also assume that $k<n$ and $n \not\equiv 1,2,\dots,\delta-1 \mod r+\delta-1$. Write $n=a(r+\delta-1)+b$ with $0 \leq b < r+\delta-1$.

We will construct a generator matrix for a linear code under above assumptions. The minimum distance of the constructed code is studied in Subsection \ref{Subsec:Analysis}.

Next we will build $A=\left\lceil\frac{n}{r+\delta-1}\right\rceil$ sets $S_1, S_2, \dots, S_A$ such that each of them consists of $r+\delta-1$ vectors from $\mathbb{F}_q^k$ except $S_A$ that consists of $n-(A-1)(r+\delta-1)$ vectors from $\mathbb{F}_q^k$. Write
\[
M=(I_r | B_{r \times (\delta-1)})=
\left( \begin{matrix}
   a_{1,1} & \hdots & a_{1,r+\delta-1} \\
   \vdots  &        & \vdots \\
   a_{r,1} & \hdots & a_{r,r+\delta-1}
  \end{matrix} \right)
\]
where $I_r$ is an identity matrix of size $r$ and $B_{r \times (\delta-1)}$ is such $r \times (\delta-1)$ matrix that all its square submatrices are invertible. Define further
\[
U_0 = \left\{ a_{i_1,i_2} \mid 1 \leq i_1 \leq r \text{ and } 1 \leq i_2 \leq r+\delta-1 \right\} 
\]
and
\[
U_{m+1} = \left\{ a - \frac{bc}{d} \mid a,b,c,d \in U_m \text{ and } d \neq 0 \right\} \cup U_m
\]
for $m=0,\dots,r$. We have $U_0 \leq r\delta$, $|U_{m+1}| \leq |U_m|^4$ and $|U_{r+1}| \leq (r\delta)^{4^{r+1}}$.

First, choose any $r$ linearly independent vectors $\mathbf{g}_{1,1},\dots,\mathbf{g}_{1,r} \in \mathbb{F}_{q}^{k}$. Let
\[
\mathbf{s}_{1,r+j}=\sum_{l=1}^{r}a_{l,r+j}\mathbf{g}_{1,l}
\]
for $j=1,\dots,\delta-1$. These $r+\delta-1$ vectors form the set $S_1$. Notice that these vectors correspond the columns of matrix
\[
(\mathbf{g}_{1,1}|\dots|\mathbf{g}_{1,r}) M.
\]
This set has the property that any $r$ vectors from this set are linearly independent.

Let $1< i \leq A$. Assume that we have $i-1$ sets $S_1, S_2, \dots, S_{i-1}$ such that when taken at most $k$ vectors from these sets, at most $r$ vectors from each set, these vectors are linearly independent. Next we will show inductively that this is possible by constructing the set $S_i$ with the same property.

Let $\mathbf{g}_{i,1}$ be any vector such that when taken at most $k-1$ vectors from the already built sets, with at most $r$ vectors from each set, then $\mathbf{g}_{i,1}$ and these $k-1$ other vectors are linearly independent. This is possible since $\binom{n}{k-1}q^{k-1}<q^k$.

Write $\mathbf{s}_{i,r+m}^{(h)} = \sum_{l=1}^{h} a_{l,r+m}\mathbf{g}_{i,l}$ for $m=1,\dots,\delta-1$ and $h=1,\dots,r$ and to shorten the notation, write $\mathbf{s}_{i,r+m}=\mathbf{s}_{i,r+m}^{(r)}$ for $m=1,\dots,\delta-1$.

Suppose we have $j-1$ vectors $\mathbf{g}_{i,1},\dots,\mathbf{g}_{i,j-1}$ such that when taken at most $k$ vectors from the sets $S_1, S_2, \dots, S_{i-1}$ or $\{\mathbf{g}_{i,1},\dots,\mathbf{g}_{i,j-1},\mathbf{s}_{i,r+1}^{(j-1)},\dots,\mathbf{s}_{i,r+\delta-1}^{(j-1)} \}$, with at most $r$ vectors from each set $S_1, S_2, \dots, S_{i-1}$ and at most $j-1$ vectors from the set $\{\mathbf{g}_{i,1},\dots,\mathbf{g}_{i,j-1},\mathbf{s}_{i,r+1}^{(j-1)},\dots,\mathbf{s}_{i,r+\delta-1}^{(j-1)} \}$, then these vectors are linearly independent.

Let
\[
V_j=\{ u_1\mathbf{g}_{i,1} + \dots + u_j\mathbf{g}_{i,j} \mid u_h \in U_{r+1} \text{ and } u_j \neq 0 \} \bigcup  \{\mathbf{g}_{i,1},\dots,\mathbf{g}_{i,j} \}.
\]
Notice that $V_j$ is finite and to be precise, $|V_j| < 2|U_{r+1}|^{j} \leq 2 (r\delta)^{j4^{r+1}}$.

Choose $\mathbf{g}_{i,j}$ to be any vector with the following properties: when taken at most $k-1$ vectors from the sets $S_1, S_2, \dots, S_{i-1}$ or $V_{j-1}$, with at most $r$ vectors from each set $S_1, S_2, \dots, S_{i-1}$ and at most $j-1$ vectors from the set $V_{j-1}$, then none of the vectors in $V_{j} \setminus \{ \mathbf{g}_{i,1},\dots,\mathbf{g}_{i,j-1} \}$ does not belong to a subspace that these $k-1$ other vectors span. This is possible because there are at most $\binom{n+2 (r\delta)^{j4^{r+1}}}{k-1}$ different possibilities to choose, each of the options span a subspace with $q^{k-1}$ vectors, and since $q$ is large we have $2(r\delta)^{j4^{r+1}+4^{r+1}}\binom{n+2 (r\delta)^{j4^{r+1}}}{k-1}q^{k-1}<q^k$. Notice that $u\mathbf{g}_{i,j}+\mathbf{v} \in V$ (where $V$ is some subspace) if and only if $u\mathbf{g}_{i,j} \in -\mathbf{v}+V$.

To prove the induction step we have to prove the following thing: when taken at most $k$ vectors from sets $S_1, S_2, \dots, S_{i-1}$ or $\{\mathbf{g}_{i,1},\dots,\mathbf{g}_{i,j},\mathbf{s}_{i,r+1}^{(j)},\dots,\mathbf{s}_{i,r+\delta-1}^{(j)}\}$, with at most $r$ vectors from each set $S_1, S_2, \dots, S_{i-1}$ and at most $j$ vectors from the set $\{\mathbf{g}_{i,1},\dots,\mathbf{g}_{i,j},\mathbf{s}_{i,r+1}^{(j)},\dots,\mathbf{s}_{i,r+\delta-1}^{(j)}\}$, then these vectors are linearly independent. Let $1 \leq l \leq j$, $\mathbf{v}$ be a sum of at most $k-l$ vectors from the sets $S_1, S_2, \dots, S_{i-1}$ with at most $r$ vectors from each set. We will assume a contrary: We have coefficients $s_{1},\dots,s_{l} \in \mathbf{F}_q \setminus \{ 0 \}$ such that:
\[
\mathbf{v} + \sum_{m=1}^{l}s_m\sum_{h=1}^{j}a_{h,f_m}\mathbf{g}_{i,h} = \mathbf{0}
\]
with $f_1 \leq \dots \leq f_l$ and $f_m \not\in \{ j+1,j+2,\dots,r \}$ for $m=1,\dots,l$.

Write
\[
\sum_{m=1}^{l}s_m\sum_{h=1}^{j}a_{h,f_m}\mathbf{g}_{i,h} = \sum_{h=1}^{j}b_h \mathbf{g}_{i,h},
\]
\emph{i.e.},
\[
\begin{pmatrix}
    b_1 \\
    \vdots \\
    b_j
\end{pmatrix}
=\begin{pmatrix}
    a_{1,f_1} & \hdots & a_{1,f_l} \\
    \vdots & & \vdots \\
    a_{j,f_1} & \hdots & a_{j,f_l}
\end{pmatrix}
\begin{pmatrix}
    s_1 \\
    \vdots \\
    s_l
\end{pmatrix}.
\]
Without loss of generality we may assume that $a_{j,f_l} \neq 0$.

Let $t$ be the smallest non-negative integer such that $b_{j-t} \neq 0$. Such $t$ exists since the rank of $(a_{h,f_i})_{j \times l}$ is $l$ and
\[
\begin{pmatrix}
    s_1 \\
    \vdots \\
    s_l
\end{pmatrix} \neq \mathbf{0}.
\]

Hence we have
\[
\begin{pmatrix}
    b_1 \\
    \vdots \\
    b_{j-t} \\
    0 \\
    \vdots \\
    0
\end{pmatrix}
=\begin{pmatrix}
    c_{1,f_1}^{(1)} & \hdots & c_{1,f_{l-1}}^{(1)} & 0 \\
    \vdots & & \vdots  & \vdots \\
    c_{j-t,f_1}^{(1)} & \hdots & c_{j-t,f_{l-1}}^{(1)} & 0 \\
    c_{j-t+1,f_1}^{(1)} & \hdots & c_{j-t+1,f_{l-1}}^{(1)} & 0 \\
    \vdots & & \vdots & \vdots \\
    c_{j-1,f_1}^{(1)} & \hdots & c_{j-1,f_{l-1}}^{(1)} & 0 \\
    a_{j,f_1} & \hdots & a_{j,f_{l-1}} & a_{j,f_l}
\end{pmatrix}
\begin{pmatrix}
    s_1 \\
    \vdots \\
    s_l
\end{pmatrix}
\]
where $c_{h,f_i}^{(1)}=a_{h,f_i}-\frac{a_{h,f_l}a_{j,f_i}}{a_{j,f_l}} \in U_1$.

This gives
\[
\begin{pmatrix}
    b_1 \\
    \vdots \\
    b_{j-t} \\
    0 \\
    \vdots \\
    0
\end{pmatrix}
=\begin{pmatrix}
    c_{1,f_1}^{(1)} & \hdots & c_{1,f_{l-1}}^{(1)} \\
    \vdots & & \vdots  \\
    c_{j-t,f_1}^{(1)} & \hdots & c_{j-t,f_{l-1}}^{(1)} \\
    c_{j-t+1,f_1}^{(1)} & \hdots & c_{j-t+1,f_{l-1}}^{(1)} \\
    \vdots & & \vdots  \\
    c_{j-1,f_1}^{(1)} & \hdots & c_{j-1,f_{l-1}}^{(1)}
\end{pmatrix}
\begin{pmatrix}
    s_1 \\
    \vdots \\
    s_{l-1}
\end{pmatrix}
=\begin{pmatrix}
    c_{1,f_1}^{(2)} & \hdots & c_{1,f_{l-2}}^{(2)} & 0 \\
    \vdots & & \vdots & \vdots  \\
    c_{j-t,f_1}^{(2)} & \hdots & c_{j-t,f_{l-2}}^{(2)} & 0 \\
    c_{j-t+1,f_1}^{(2)} & \hdots & c_{j-t+1,f_{l-2}}^{(2)} & 0 \\
    \vdots & & \vdots & \vdots  \\
    c_{j-2,f_1}^{(1)} & \hdots & c_{j-2,f_{l-2}}^{(1)} & 0 \\
    c_{j-1,f_1}^{(1)} & \hdots & c_{j-1,f_{l-2}}^{(1)} & c_{j-1,f_{l-1}}^{(1)}
\end{pmatrix}
\begin{pmatrix}
    s_1 \\
    \vdots \\
    s_{l-1}
\end{pmatrix}
\]
and by continuing the process
\[
\begin{pmatrix}
    b_1 \\
    \vdots \\
    b_{j-t} \\
\end{pmatrix}
=\begin{pmatrix}
    c_{1,f_1}^{(t)} & \hdots & c_{1,f_{l-t}}^{(t)} \\
    \vdots & & \vdots \\
    c_{j-t,f_{1}}^{(t)} & \hdots & c_{j-t,f_{l-t}}^{(t)} \\
\end{pmatrix}
\begin{pmatrix}
    s_1 \\
    \vdots \\
    s_{l-t}
\end{pmatrix}
\]
where $c_{h,f_i}^{(v)}=c_{h,f_i}^{(v-1)}-\frac{c_{h,f_{l-v+1}}^{(v-1)}c_{j-v+1,f_{i}}^{(v-1)}}{c_{j-v+1,f_{l-v+1}}^{(v-1)}} \in U_v$ for $2 \leq v \leq t$.

We can continue the process (\emph{i.e.}, $c_{j-v+1,f_{l-v+1}}^{(v-1)} \neq 0$) since the smallest non-invertible square matrix in the right lower corner of
\[
\begin{pmatrix}
    a_{1,f_1} & \hdots & a_{1,f_l} \\
    \vdots & & \vdots \\
    a_{j,f_1} & \hdots & a_{j,f_l}
\end{pmatrix}
\]
has the side length at least $t+2$, if even exist. Indeed, suppose that matrices in the right lower corner with side length less than or equal to $N$ are invertible and $N$ is maximal. The value $N$ is well-defined since the square matrix with side length $1$ is invertible. Assume contrary: $N \leq t$ and write
\[
C=\begin{pmatrix}
    a_{j-N+1,f_1} & \hdots & a_{j-N+1,f_{l-N}} \\
    \vdots & & \vdots \\
    a_{j,f_1} & \hdots & a_{j,f_{l-N}}
\end{pmatrix}.
\]

Assume first that $C$ is a zero matrix. Now
\[
\mathbf{0}=\begin{pmatrix}
    a_{j-N+1,f_{l-N+1}} & \hdots & a_{j-N+1,f_{l}} \\
    \vdots & & \vdots \\
    a_{j,f_{l-N+1}} & \hdots & a_{j,f_{l}}
\end{pmatrix}
\begin{pmatrix}
    s_{l-N+1} \\
    \vdots \\
    s_{l}
\end{pmatrix}
\]
that is not possible.

Assume then that $C$ is not a zero matrix. Clearly $N$ is greater than or equal to the number of columns in
\[
\begin{pmatrix}
    a_{1,f_1} & \hdots & a_{1,f_l} \\
    \vdots & & \vdots \\
    a_{j,f_1} & \hdots & a_{j,f_l}
\end{pmatrix}
\]
corresponding the columns of $B_{r \times (\delta-1)}$. Hence
\[
\begin{pmatrix}
    a_{j-N,f_{l-N}} & \hdots & a_{j-N,f_{l}} \\
    \vdots & & \vdots \\
    a_{j,f_{l-N}} & \hdots & a_{j,f_{l}}
\end{pmatrix}
=(\mathbf{e}_1|\mathbf{e}_2|\dots|\mathbf{e}_{\epsilon}|B')
\]
where each $\mathbf{e}_i$ has one $1$ and other elements are zeros, and these $1$s are in different rows, and all the square submatrices of $B'$ are invertible. Hence it is also invertible, against assumption. This proves that $N \geq t+1$.

Hence also $c_{j-t,f_{l-t}}^{(t)} \neq 0$, and we have
\begin{equation}
\begin{split}
\begin{pmatrix}
    b_1 \\
    \vdots \\
    b_{j-t} \\
\end{pmatrix}
= & \begin{pmatrix}
    c_{1,f_1}^{(t)}-\frac{c_{j-t,f_{1}}^{(t)}c_{1,f_{l-t}}^{(t)}}{c_{j-t,f_{l-t}}^{(t)}} & \hdots & c_{1,f_{l-t-1}}^{(t)}-\frac{c_{j-t,f_{l-t-1}}^{(t)}c_{1,f_{l-t}}^{(t)}}{c_{j-t,f_{l-t}}^{(t)}} & c_{1,f_{l-t}}^{(t)} \\
    \vdots & & \vdots & \vdots \\
    c_{j-t-1,f_1}^{(t)}-\frac{c_{j-t,f_{1}}^{(t)}c_{j-t-1,f_{l-t}}^{(t)}}{c_{j-t,f_{l-t}}^{(t)}} & \hdots & c_{j-t-1,f_{l-t-1}}^{(t)}-\frac{c_{j-t,f_{l-t-1}}^{(t)}c_{j-t-1,f_{l-t}}^{(t)}}{c_{j-t,f_{l-t}}^{(t)}} & c_{j-t-1,f_{l-t}}^{(t)} \\
    0 & \hdots & 0 & c_{j-t,f_{l-t}}^{(t)} \\
\end{pmatrix} \\
& \cdot \begin{pmatrix}
    s_1 \\
    \vdots \\
    s_{l-t-1} \\
    \frac{s_{1}c_{j-t,f_{1}}^{(t)} + \dots + s_{l-t}c_{j-t,f_{l-t}}^{(t)}}{c_{j-t,f_{l-t}}^{(t)}}
\end{pmatrix}
\end{split}
\end{equation}
and hence
\begin{equation}
\begin{split}
\mathbf{0} & = \mathbf{v} + \sum_{m=1}^{l}s_m\sum_{h=1}^{j}a_{h,f_m}\mathbf{g}_{i,h} \\
 &= \mathbf{v} + \sum_{m=1}^{l-t-1}s_m\sum_{h=1}^{j-t-1}\left(c_{h,f_m}^{(t)}-\frac{c_{j-t,f_{m}}^{(t)}c_{h,f_{l-t}}^{(t)}}{c_{j-t,f_{l-t}}^{(t)}}\right)\mathbf{g}_{i,h} + \frac{s_1c_{j-t,f_{1}}^{(t)}+\dots+s_{l-t}c_{j-t,f_{l-t}}^{(t)}}{c_{j-t,f_{l-t}}^{(t)}}\sum_{h=1}^{j-t}c_{h,f_{l-t}}^{(t)}\mathbf{g}_{i,h}
\end{split}
\end{equation}
which cannot be true since $(l-t-1)+1 \leq j-t$ and $\sum_{h=1}^{j-t}c_{h,f_{l-t}}^{(t)}\mathbf{g}_{i,h}$ is chosen such that it does not belong to the subspace that $\mathbf{v}, \sum_{h=1}^{j-t-1}\left(c_{h,f_1}^{(t)}-\frac{c_{j-t,f_{1}}^{(t)}c_{h,f_{l-t}}^{(t)}}{c_{j-t,f_{l-t}}^{(t)}}\right)\mathbf{g}_{i,h},\dots, \sum_{h=1}^{j-t-1}\left(c_{h,f_{l-t-1}}^{(t)}-\frac{c_{j-t,f_{l-t-1}}^{(t)}c_{h,f_{l-t}}^{(t)}}{c_{j-t,f_{l-t}}^{(t)}}\right)\mathbf{g}_{i,h}$ span, and $\frac{s_1c_{j-t,f_{1}}^{(t)}+\dots+s_{l-t}c_{j-t,f_{l-t}}^{(t)}}{c_{j-t,f_{l-t}}^{(t)}} \neq 0$ since $b_{j-t} \neq 0$.

Now, the sets $S_i$ consist of vectors $\{\mathbf{g}_{i,1},\dots,\mathbf{g}_{i,r},\mathbf{s}_{i,r+1},\dots,\mathbf{s}_{i,r+\delta-1}\}$ for $i=1,\dots,a$. If $b \neq 0$ the set $S_A$ consists of vectors $\{\mathbf{g}_{A,1},\dots,\mathbf{g}_{A,b-\delta+1},\mathbf{s}_{i,r+1}^{(b-\delta+1)},\dots,\mathbf{s}_{i,r+\delta-1}^{(b-\delta+1)}\}$. The matrix $\mathbf{G}$ is a matrix with vectors from the sets $S_1,S_2,\dots,S_A$ as its column vectors, \emph{i.e.},
\[
\mathbf{G}=\left(\mathbf{G}_1|\mathbf{G}_2|\dots|\mathbf{G}_A\right)
\]
where
\[
\mathbf{G_j}=\left(\mathbf{g}_{j,1}|\dots|\mathbf{g}_{j,r}|\mathbf{s}_{i,r+1}|\dots|\mathbf{s}_{i,r+\delta-1}\right)
\]
for $i=1,\dots,a$, and
\[
\mathbf{G_A}=\left(\mathbf{g}_{A,1}|\dots|\mathbf{g}_{A,b-\delta+1}|\mathbf{s}_{i,r+1}^{(b-\delta+1)}|\dots|\mathbf{s}_{i,r+\delta-1}^{(b-\delta+1)}\right)
\]
if $b\neq 0$.

To be a generator matrix for a code of dimension $k$, the rank of $\mathbf{G}$ has to be $k$. By the construction the rank is $k$ if and only if
$n-A(\delta-1) \geq k$, and this is what we assumed.

\begin{remark}
Notice that the estimations for $q$ are very rough in the construction. This is because we are mainly interested in the randomized case in which $q \rightarrow \infty$.
\end{remark}

\begin{remark}
Notice that in the above construction we could have chosen different matrices $M=(I_r | B_{r \times (\delta-1)})$ for each $\mathbf{G_j}$. Also, the sets $S_j$ do not have to be of the given size. We only need to assume that
\[
\sum_{j} |S_j| =n
\]
and $\delta \leq |S_j| \leq r +\delta-1$. Then the corresponding matrix is of type $(I_{|S_j|-\delta+1} | B_{(|S_j|-\delta+1) \times (\delta-1)})$.
\end{remark}

\subsection{Construction with random vectors}\label{Subsec:RandomProof}
If we choose randomly the vector $\mathbf{g}_{i,j}$ in the above construction the probability that we get a nonsuitable choice is at most
\[
2(r\delta)^{r4^{r+1}+4^{r+1}}\binom{n+2 (r\delta)^{r4^{r+1}}}{k-1}q^{k-1}.
\]
The size of the whole vector space is $q^k$. Hence the probability that the whole code is as in our construction, is at least
\[
\left(\frac{q^k-2(r\delta)^{r4^{r+1}+4^{r+1}}\binom{n+2 (r\delta)^{r4^{r+1}}}{k-1}q^{k-1}}{q^k}\right)^n=\left(1-\frac{2(r\delta)^{r4^{r+1}+4^{r+1}}\binom{n+2 (r\delta)^{r4^{r+1}}}{k-1}}{q}\right)^n \rightarrow (1-0)^n=1,
\]
as $q \rightarrow \infty$.

\subsection{The minimum distance of the constructed code}\label{Subsec:Analysis}
Next we will calculate the minimum distance of the constructed code with the assumption that the sets $S_j$ are of size $s_j$ ($j=1,\dots,A$), respectively. Assume also that $s_1 \leq \dots \leq s_A$. Write
\[
\mathbf{G}=\left(E_1|F_1|E_2|F_2|\dots|E_A|F_A\right),
\]
where $E_j=\left( \mathbf{g}_{j,1}|\dots|\mathbf{g}_{j,s_j-\delta+1} \right)$ and $F_j=\left( \mathbf{s}_{j,r+1},\dots,\mathbf{s}_{j,r+\delta-1} \right)$ for $j=1,\dots,A$.

Let $e_1,\dots,e_k \in \mathbb{F}_q$ be such elements that
\[
(e_1,\dots,e_k)\mathbf{G}
\]
is a vector of minimal weight. By changing columns between $E_j$s and $F_j$s we may assume that the weight of
\[
(e_1,\dots,e_k)\left(E_1|E_2|\dots|E_A\right)
\]
is minimal, that is, it has the biggest possible amount of zeros. So it has $k-1$ zeros.

Suppose that
\[
(e_1,\dots,e_k)F_j
\]
has a zero, \emph{i.e.}, its weight is not $\delta-1$. If $(e_1,\dots,e_k)E_j \neq \mathbf{0}$, then by changing columns between $E_j$ and $F_j$ we would get one more zero into $(e_1,\dots,e_k)\left(E_1|E_2|\dots|E_A\right)$ and that is not possible. Hence the number of zeros in $(e_1,\dots,e_k)\left(F_1|F_2|\dots|F_A\right)$ is at most $z(\delta-1)$ where $z$ is an integer with properties
\[
\sum_{j=1}^{z} (s_j-\delta+1) \leq k-1 \text{ and } \sum_{j=1}^{z+1} (s_j-\delta+1) > k-1.
\]
Hence the minimum distance of the code is
\[
n-(k-1)-z(\delta-1).
\]

\begin{exam}
Suppose that $n=A(r+\delta-1)$ and choose that $s_j-\delta+1=r$ for all $j=1,\dots,A$. Then, $z=\left\lfloor\frac{k-1}{r}\right\rfloor$ and hence the minimum distance is
\[
n-(k-1)-\left\lfloor\frac{k-1}{r}\right\rfloor(\delta-1)
=n-k-\left(\left\lceil\frac{k}{r}\right\rceil-1\right)(\delta-1) + 1 = d_{\text{opt}}(n,k,r,\delta)
\]
and hence the construction is optimal.

Suppose then that $n=a(r+\delta-1)+b$ with $0 \leq b<r+\delta-1$. If $0 < b < \delta$ then using the above optimal code with extra $b$ zero columns in the generator matrix we get a code with minimum distance $d_{\text{opt}}(n-b,k,r,\delta)=d_{\text{opt}}(n,k,r,\delta)-b$.

If $b \geq \delta$ then choose $s_j=r+\delta-1$ for $j=1,\dots,a$ and $s_{a+1}=b$. Now $z=\left\lceil\frac{k-b+\delta-1}{r}\right\rceil$ and hence the minimum distance is
\[
n-k-\left\lceil\frac{k-b+\delta-1}{r}\right\rceil(\delta-1) + 1.
\]
The distance $d_{\text{opt}}(n,k,r,\delta)-\left(n-k-\left\lceil\frac{k-b+\delta-1}{r}\right\rceil(\delta-1) + 1\right)$
is
\[
(\delta-1)\left(\left\lceil\frac{k-b+\delta-1}{r}\right\rceil-\left\lceil\frac{k}{r}\right\rceil +1\right) \leq \delta-1
\]
and hence the code is again at least almost optimal.
\end{exam}

\section{Conclusion}
In this paper we have studied linear locally repairable codes with all-symbol locality. We have constructed codes with almost optimal minimum distance. Namely, the difference between largest achievable minimum distance of locally repairable codes and the minimum distance of our codes is maximally $\delta-1$. Instead of just giving a construction, it is shown that by using random matrices with a guaranteed locality property, such matrix generates an almost optimal LRC with probability approaching to one as the field size approaches to infinity.

Also, methods to build new codes for different parameters using already existing codes are presented. Namely, a method to find a bigger code and a method to find a smaller code are presented.

As a future work it is still left to find the exact expression of the largest achievable minimum distance of the linear locally repairable code with all-symbol locality when given the length $n$ and the dimension $k$ of the code and the locality $(r,\delta)$. 


\end{document}